\documentclass[11pt]{article}
\usepackage[utf8]{inputenc}
\usepackage{mathtools}
\usepackage{comment}
\usepackage{amsmath}
\usepackage{algorithm}
\usepackage{algpseudocode}
\usepackage{amsfonts,amsmath,amssymb,amsthm,boxedminipage,color,url,fullpage}
\usepackage[numbers]{natbib}
\usepackage[pdfstartview=FitH,colorlinks,linkcolor=blue,filecolor=blue,citecolor=blue,urlcolor=blue]{hyperref}
\usepackage[labelfont=bf]{caption}
\usepackage{aliascnt}
\usepackage{authblk}
\usepackage{tikz}
\usepackage{accents}
\usepackage{pgfplots}
\usepackage[capitalise]{cleveref}
\pgfplotsset{width=10cm,compat=1.9}
\usepackage{soul}
\setstcolor{red}

\newtheorem{theorem}{Theorem}
\newtheorem{lemma}{Lemma}
\newtheorem{corollary}{Corollary}
\newtheorem{claim}{Claim}
\newtheorem{observation}{Observation}
\newtheorem{proposition}{Proposition}

\newtheorem{example}{Example}

\newcommand{\eps}{\varepsilon}
\newcommand{\demand}{\mathcal{D}}
\newcommand{\critical}{\mathcal{C}}
\newcommand{\reals}{\mathbb{R}}

\DeclareMathOperator*{\argmax}{arg\,max}

\newcommand{\dmin}{{\delta_{\min}}}
\newcommand{\dmax}{{\delta_{\max}}}
\def\multiset#1#2{\ensuremath{\left(\kern-.3em\left(\genfrac{}{}{0pt}{}{#1}{#2}\right)\kern-.3em\right)}}
\definecolor{darkpastelgreen}{rgb}{0.01, 0.75, 0.24}
\definecolor{darkorchid}{rgb}{0.6, 0.2, 0.8}
\definecolor{blue(pigment)}{rgb}{0.2, 0.2, 0.6}

\title{Combinatorial Contracts Beyond Gross Substitutes\thanks{In concurrent work, Deo-Campo Vuong et al.~\cite{VuongEtAl23} also present an algorithm for enumerating all critical values, and use it to obtain an efficient algorithm for finding the optimal contract for supermodular rewards and additive costs. Their work was uploaded to arXiv on August 14, 2023. 
This project has received funding from the European Research Council (ERC) under the European Union's Horizon 2020 research and innovation program (grant agreement No. 866132), by the Israel Science Foundation (grant number 317/17), by an Amazon Research Award, and by the NSF-BSF (grant number 2020788).}}
\author{
Paul Dütting \thanks{Google Research, Zurich, Switzerland. Email: \texttt{duetting@google.com}} 
\;\;\; 
Michal Feldman
\thanks{Tel Aviv University, Israel. Email: \texttt{mfeldman@tauex.tau.ac.il}}
\;\;\;
Yoav Gal Tzur
\thanks{Tel Aviv University, Israel. Email: \texttt{yoavgaltzur@mail.tau.ac.il}
}
}
\date{}

\begin{document}

\maketitle

\begin{abstract}
We study the combinatorial contracting problem of D\"utting et al.~\cite{dutting2022combinatorial}, in which a principal seeks to incentivize an agent to take a set of costly actions. In their model, there is a binary outcome (the agent can succeed or fail), and the success probability and the costs depend on the set of actions taken. The optimal contract is linear, paying the agent an $\alpha$ fraction of the reward. 
For gross substitutes (GS) rewards and additive costs, they give a poly-time algorithm for finding the optimal contract.
They use the properties of GS functions to argue that there are poly-many  ``critical values" of $\alpha$, and that one can iterate through all of them efficiently in order to find the optimal contract.
 
In this work we study to which extent GS rewards and additive costs constitute a tractability frontier for combinatorial contracts. 
We present an algorithm that for {\em any} rewards and costs, enumerates all critical values, with poly-many demand queries (in the number of critical values). This implies the tractability of the optimal contract for any setting with poly-many critical values  and efficient demand oracle. A direct corollary is a poly-time algorithm for the optimal contract in settings with supermodular rewards and submodular costs.
We also study a natural class of matching-based instances with XOS rewards and additive costs. 
While the demand problem for this setting is tractable, we show that it admits an exponential number of critical values. On the positive side, we present (pseudo-) polynomial-time algorithms for two natural special cases of this setting. Our work unveils a profound connection to sensitivity analysis, and designates matching-based instances as a crucial focal point for gaining a deeper understanding of combinatorial contract settings.
\end{abstract}

\section{Introduction}

Contract theory is one of the pillars of microeconomic theory. In its most prominent model of Holmström \cite{holmstrom1979moral} and Grossman and Hart \cite{grossman1992analysis}, a principal seeks to delegate the execution of a task to an agent. 
The agent may perform one of $n$ costly actions $A = \{1,...,n\}$, unobservable by the principal. Each action leads to a stochastic reward (outcome) for the principal. 
The basic incentive problem is that the agent, who bears the cost of executing the task, does not directly benefit from the stochastic reward that is generated.
To incentivize the agent to exert effort, the principal designs a contract that specifies the payments for each outcome.
This induces a game between the principal and the agent. Given a contract, the agent replies with an action that maximizes payment minus cost. The principal in turn tries to find the contract that maximizes reward minus payment, given the agent's best response.

In the vanilla contract problem, where the agent may only choose a single action, the optimal contract can be computed in polynomial time:
This is done by computing for each action individually, using a linear program, the minimum payment required to incentivize this action (if at all possible), then pick the best action among those \cite{grossman1992analysis}.

In many real-life settings, the 
contract problem is combinatorial in one or more dimensions of the problem. 
In such settings, the approach above might not yield an efficient algorithm. 
A recent line of work embarks on the design of efficient algorithms for combinatorial settings \cite{duetting2022multi,dutting2022combinatorial, babaioff2006combinatorial, dutting2021complexity}.

\paragraph{The Model.}
We study the model introduced by \cite{dutting2022combinatorial} where the principal seeks to hire an agent to run a project, and the agent can perform \textit{any subset} of actions from a fixed set $A = \{1,2,...,n\}$.
The project can either succeed or fail. If the project succeeds, the principal gets a reward, which we normalize to be $1$, otherwise, the reward is $0$.
Each subset of actions $S$ induces a success probability of the project, specified by a set function $f: 2^A \to [0,1]$. Since $f(S)$ is also the expected reward of a set of actions, we also refer to $f$ as the reward function. The cost of a set is given by a set function
$c:2^A \to \reals_{\ge 0}$.
We assume to have value oracle access to both $f$ and $c$.

For this binary setting, the optimal contract takes a linear form, where the principal pays a fraction $\alpha \in [0,1]$ of her reward if the project succeeds and nothing otherwise \cite{carroll2015robustness}.
In this case, the agent picks a set of actions $S \subseteq A$ which maximizes $\alpha f(S) - c(S)$, and the principal chooses a contract $\alpha$ such that $(1-\alpha) f(S)$ is maximized --- with $S$ being the agent's best response to contract $\alpha$.
Motivated by a connection to combinatorial auctions, interpreting $f(S)$ as the value for a set of items and $\frac{1}{\alpha}c(S)$ as (bundle) prices, we will also refer to the agent's best response for a given $\alpha$, as the agent's \emph{demand}. 

Values of $\alpha$ at which the agent's demand changes are called \emph{critical values}, and it is easy to observe that the optimal contract $\alpha$ occurs at a critical value.

\paragraph{State of the Art.}

The main result of \cite{dutting2022combinatorial} 
is a polynomial time algorithm when the success probability function $f$ is gross substitutes (GS), and the cost function $c$ is additive.

To prove this result the authors of \cite{dutting2022combinatorial}  show that for GS $f$ and additive $c$ there can be at most $O(n^2)$ critical values of $\alpha$. In addition, they exploit that for GS valuations and item prices the agent's demand problem can be solved greedily \cite{PaesLeme17}, to obtain an efficient algorithm that given a critical value of $\alpha$ finds the next critical value $\alpha' > \alpha$. Thus, the optimal contract can be solved by iterating efficiently through critical values of $\alpha$, and choosing the value of $\alpha$ that maximizes the principal's utility. 
Crucially, the properties of GS are used both to provide an upper bound on the number of critical values, and to iterate efficiently through these values.

For the case in which $f$ is submodular and $c$ is additive, the authors of \cite{dutting2022combinatorial} show that there may be exponentially many critical values of $\alpha$. Furthermore, they show that computing the optimal contract for this case is NP-hard.

\subsection{Our Contribution}

In this work, we explore to which extent gross substitutes success probability $f$ and additive costs $c$ constitutes a tractability frontier for combinatorial contracts.

\paragraph{Enumerating All Critical Values with a Demand Oracle.}

As a first result, we present an algorithm that for \emph{any} $f$ and $c$ finds all critical values of $\alpha$, using polynomially (in the number of critical values) many calls to an algorithm for the agent's demand problem.

\vspace{0.1in}
\noindent \textbf{Proposition} (\Cref{prop:optContractSuff})\textbf{.} Given oracle access to the agent's demand, there exists an algorithm which returns all the critical values. The running time of this algorithm is linear in the number of critical values.
\vspace{0.1in}

An immediate algorithmic implication of this result is an efficient algorithm for the optimal contract problem in scenarios with a small (i.e., polynomial) number of critical values, given access to a demand oracle. In scenarios with a poly-time algorithm for a demand query, this gives a poly-time algorithm for the optimal contract.

Our algorithm is recursive. It finds the critical values by querying the agent's demand oracle at two end points of a given segment $[\alpha, \alpha'] \subseteq [0,1]$. If the agent's best response for the two contracts is the same, i.e., $S_\alpha = S_{\alpha'}$, there are no critical values in $(\alpha,\alpha']$.
Otherwise, the oracle answers differ, i.e., $S_\alpha \ne S_{\alpha'}$, and the procedure is applied recursively for the sub-segments $[\alpha, \gamma]$ and $[\gamma,\alpha']$, where $\gamma$ is the contract for which the agent is indifferent between $S_\alpha$ and $S_{\alpha'}$. That is, $\gamma$ is such that $\gamma f(S_\alpha) - c(S_\alpha) = \gamma f(S_{\alpha'}) - c(S_{\alpha'})$.

This algorithm has been discovered before in different contexts, in particular in the field of sensitivity analysis of combinatorial optimization problems \cite{gusfield1980sensitivity}, where it is referred to as the Eisner-Severance technique \cite{eisner1976mathematical}.
We use this algorithm, and the connection to sensitivity analysis more generally, to make progress on the tractability frontier for optimal contracts beyond gross substitutes rewards and additive costs.

\paragraph{Complementary Actions.}
A first setting for which we show the tractability of the optimal contract is the case of supermodular $f$ and submodular $c$. This captures complementary actions with non-decreasing marginal rewards and non-increasing marginal costs.

Complementary actions appear in many real-world applications, specifically when the agent can perform actions that are infrastructural in nature (e.g., buying a computer, learn a new technology) and the infrastructure can be used for various purposes.

In this setting, the agent's demand problem can be solved efficiently using value queries (via reduction to submodular minimization). 
We show that the number of critical values in this setting is at most $n$. In combination with \Cref{prop:optContractSuff} we thus obtain:

\vspace{0.1in}
\noindent \textbf{Theorem} (\Cref{thm:submodular-costs})\textbf{.} If the success probability function is monotone supermodular, and the cost function is monotone submodular, the optimal contract can be computed in polynomial time.

\paragraph{Matching-based Rewards and Costs.}

In the second setting that we consider, actions are edges in a bipartite graph $G = (V \cup U, A)$. 
We think of the right-hand side vertices $u \in U$ as resources, and the left-hand side vertices $v \in V$ as tasks. 
Every matching in $G$ is associated with a success probability, and the success probability for a set of edges $S$ is the maximal success probability of a matching contained in $S$. Similarly, each set of edges $S$ (independent of whether it forms a matching or not) is associated with a cost.  

A natural starting point, and the focus of this work, is where each edge is associated with a cost $c_{v,u}$ and a success probability $f_{v,u}$. The cost for a set of edges $S$ is given by the sum of costs $c_{v,u}$ of $(v,u) \in S$. The success probability of a set of edges $S$ is the max-weight matching in the subgraph induced by $S$, with respect to weights $\{f_{v,u}\}_{(v,u)\in A}$. 
In this case, the success probability $f(S)$ is XOS, but not necessarily submodular.\footnote{For each matching $M 
\subseteq A$ in the graph $G$, let $x_M \in \mathbb{R}_{\ge 0}^{|A|}$ be its corresponding vector, with $x_M(a) = f_a$ if $a \in M$ and $0$ otherwise. Then $f$ is XOS, as $f(S)=
max_{M} x_M \cdot \chi_S$, where $\chi_S$ is the characteristic vector of the set $S$.}

The agent's best response in this setting corresponds to a max-weight matching computation, so a demand oracle can be computed efficiently. 
However, as our main technical result, we show a super-polynomial lower bound on the number of critical values. 
We prove this lower bound through a non-trivial reduction from the parametric shortest $s-t$ path problem in sensitivity analysis, and a family of hard instances due to Carstensen \cite{carstensen1983complexity} and Mulmuley and Shah \cite{mulmuley2000lower}, to an instance of the agent's best response problem.

\vspace{0.1in}
\noindent \textbf{Theorem} (\Cref{thm:superPolyCVs})\textbf{.}
For matching based rewards and additive costs, 
the number of critical values may be super-polynomial.
\vspace{0.1in}

We also show tractability of the optimal contract for two important special cases. 
First, we utilize the work of Carstensen \cite{carstensen1983complexity} to derive a quasi-polynomial upper bound on the number of critical values where costs and rewards are integers (up to scaling required for normalization).

\vspace{0.1in}
\noindent \textbf{Proposition} 
(\Cref{cor:MathcingQuasiPoly})\textbf{.}
For matching based rewards and additive costs, when costs and rewards are integers,  
there is a quasi-polynomial time algorithm for finding the optimal contract.
\vspace{0.1in}

We also give a polynomial-time algorithm for the case where the cost of each edge depends only on one side of the graph, i.e., 
all edges incident to the same vertex in, say $U$, have the same cost.
We prove this result by a non-immediate reduction to the case of gross substitutes $f$ \cite{dutting2022combinatorial}. 
The crux of this argument is that while $f$ when viewed as a function of a subset of edges is not gross substitutes (and not even submodular), in this special case it's possible to reformulate the agent's utility problem in terms of the vertices of $U$ with additive costs and a Rado (and hence gross substitutes) reward function (see \Cref{sec:model} for definitions).

\vspace{0.1in}
\noindent \textbf{Theorem} 
(\Cref{thm:oneSidedCosts})\textbf{.}
For matching based rewards and additive costs when the costs depend only on one side, the optimal contract can be computed in polynomial time.
\vspace{0.1in}

Our work unveils a profound connection between contracts and sensitivity analysis. It also designates matching-based instances as a pivotal focal point beyond gross substitutes (and even beyond submodular) for further research in this field.  
This is because they admit an efficient demand oracle, but at the same time they admit exponentially many critical values (as we show in \Cref{thm:superPolyCVs}). Our work thus opens up two exciting directions for future work: either (1) finding a poly-time algorithm for computing the optimal contract (which would inevitably have to sidestep enumerating all critical values), or (2) showing a hardness result (which would necessarily not rely on the hardness of demand queries). Either of these solutions would contribute to a deeper comprehension of the overall landscape.

\subsection{Related Work}
The classic hidden-action principal agent problem was formulated by Holmstr\"om \cite{holmstrom1979moral} and Grossman and Hart \cite{grossman1992analysis}, who also gave a poly-time algorithm for finding the optimal contract based on linear programming. For the types of structured contracting settings as we study them in this work, this approach however does not yield a poly-time algorithm. 
Holmstr\"om and Milgrom \cite{holmstrom1991multitask} consider a dynamic contracting setting, and show that linear contracts are optimal in that setting.
In more recent work, Carroll \cite{carroll2015robustness} shows the (max-min) robust optimality of linear contracts. In his model, the principal knows a fixed subset of the agent's actions and show that the max-min optimal contract is linear.
Dütting et al. \cite{dutting2019simple} show the max-min optimality of linear contracts also hold when the principal is aware of the agent's actions, but only the expected value of the outcome distributions are known. They also discuss how well linear contracts approximate the optimal one. In our model with binary outcomes linear contract are optimal.

 An exponential blow-up in complexity also naturally arises when the principal seeks to hire a \textit{team of agents}. Babaioff et al. \cite{babaioff2006combinatorial} introduced the combinatorial agency model, in which a team of agents, each chooses whether to exert effort or note, is incentivized by the principal. The outcome depends on the different combination of effort-making agents, whose structure is also expressed by a set function.
In later work Babioff et al. \cite{babaioff2009free}, \cite{babaioff2006mixed} study the effects of mixed strategy and free-riding in this setting.
These earlier works focused on the case where the success probability is encoded as a read-once network. The recent work of D\"utting et al. \cite{duetting2022multi} studies complement free reward functions -- XOS, submodular and subadditive. 
The key result of that work are poly-time algorithms for XOS and submodular rewards that yield $O(1)$-approximations to the optimal contract with poly-many demand and value oracle calls.
A related but different multi-agent contracting problem was recently studied by Castiglioni et al.~\cite{CastiglioniM023}, who consider the case where the principal can observe each agent's individual outcome (and base her contracts on this information).

The exponential blow-up can stem also from a \emph{structured outcome space}, which is studied by D\"utting et al. \cite{dutting2021complexity}. 
They consider scenarios with a larger outcome space (of size $m$), whose description size is $\textsf{poly}(n,\log(m))$. The goal is to find an optimal or near-optimal contract in time polynomial in the scuccinct description of the instance.

Ho et al. \cite{ho2014adaptive}, Zhu et al.~\cite{ZhuBYWJJ23}, and D\"utting et al.~\cite{DuettingGuruganeshSchneiderWang23} consider a contract design question from an online learning perspective. Kleinberg and Kleinberg \cite{kleinberg2018delegated} and Kleinberg and Raghavan \cite{kleinberg2020classifiers} consider problems that can
be thought of as contract design without money.
Guruganesh et al. \cite{guruganesh2021contracts,GuruganeshSW023}, Castiglioni et al. \cite{CastiglioniM021,castiglioni2022designing} and Alon et al. \cite{alon2021contracts,AlonDLT23} consider settings in which screening (that is, hidden types) is combined with moral hazard (hidden actions).

\section{Model and Notation}\label{sec:model}
We consider a principal-agent setting, where the agent can choose any subset of a given action set $A = \{1,2,...,n\}$, and the outcome is binary (success or failure).
A success outcome gives the principal a reward $r \in \reals_{\geq 0}$ (and a failure gives the principal zero reward).
Every set of actions $S \subseteq A$ is associated with a cost (incurred by the agent) and a success probability, which are given by set functions $c:2^A \to \reals_{\ge 0}$ and $f:2^A \to [0,1]$, respectively. 

The principal cannot observe the agent's actions, only the final outcome.
Thus, the principal incentivizes the agent to take costly actions by offering a contract $t: \{0,1\} \to \reals_{\ge0}$, which specifies a payment for every outcome. 
As standard in the literature, we impose a limited liability constraint, meaning that payments go only from the principal to the agent.
The agent observes the contract $t$, and responds by taking a set of actions $S$.

The principal's goal --- our goal in this paper --- is to find the contract that maximizes her utility, given that the agent responds with the action set that maximizes the agent's utility. This contract is referred to as the \emph{optimal contract}. 

In this binary outcome setting, it is without loss of generality to assume that $t(0)=0$  \cite{dutting2022combinatorial}.~Furthermore, we normalize the reward in case of a success to $1$, i.e., $r=1$, so $f(S)$ is the expected reward for a set of actions $S$. In this normalized setting, we often refer to $f$ as the \emph{reward} function.
A contract can now be described using a single parameter $\alpha \in [0,1]$, which is the fraction of the reward ($r=1$) that is paid to the agent, that is, $t(1) = \alpha$. 
The agent and principal's utilities are then given by:
\[
u_a(\alpha, S) = \alpha f(S) - c(S) \quad \text{and} \quad
u_p(\alpha, S) = (1-\alpha)f(S).
\]
Given a contract $\alpha$, the agent's \emph{best response}, is a set of actions $S^* \subseteq A$ 
that maximizes $u_a(\alpha, S)$ among all $S \subseteq A$. 
As standard in the literature, we assume that the agent breaks ties in favor of the principal.
Namely, for any two sets $S,S'$ such that $u(\alpha,S)=u(\alpha,S')$, the agent prefers the set with the higher value of $f$.

Drawing from the combinatorial auctions literature, we refer to the set of sets that maximize the agent's utility under a given contract $\alpha$ as the agent's \emph{demand} under this contract. By employing a consistent tie-breaking rule, we can narrow this collection down to a single set, $S_\alpha$. 
We can now write the agent and principal's utilities under a contract $\alpha$ as $u_a(\alpha) = u_a(\alpha, S_\alpha)$ and $u_p(\alpha) = u_p(\alpha, S_\alpha)$.

\paragraph{Classes of Set Functions.}
For any set function $f:2^A \to \reals$, we denote by $f(a \mid S)$ the \emph{marginal value} of $a \in A$ given the set $S \subseteq A$, that is $f(a\mid S) = f(S \cup \{a\}) - f(S)$.
We consider the following set functions.
\begin{itemize}
    \item $f$ is \emph{supermodular} if for every two sets $S,T \subseteq A$ such that $S\subseteq T$ and every $a \in A$, $f(a \mid S) \le f(a \mid T)$.

    \item $f$ is \emph{XOS} if it is a maximum over additive function, i.e. there exists $f_1,...,f_k \in \reals^{|A|}$ such that for any $S \subseteq A$, $f(S) = \max_i f_i\cdot \chi_S$, where $\chi_S \in \{0,1\}^{|A|}$ is the characteristic vector of $S$.
    
    \item $f$ is \emph{submodular} if for every two sets $S,T \subseteq A$ such that $S\subseteq T$ and every $a \in A$, $f(a \mid S)\ge f(a \mid T)$. Note that $f$ is submodular if and only if $-f$ is supermodular.
    
    \item $f$ is \emph{gross-substitutes (GS)} if for any two price vector $p,q \in \reals_{\ge 0}^{|A|}$, with $p \le q$ (point-wise), For every set $S \in \demand(p)$ there exists a set $T \in \demand(q)$ such that $T$ contains every action $a$ such that $p_a=q_a$. Where $\demand(p) = \max_{S\subseteq A} f(S) - p\cdot \chi_S$.\footnote{In auction literature this is known as the demand for $f$ under price $p$.} 
    
    \item  $f:2^U \to \reals_{\ge 0}$ is \emph{Rado} if there is a bipartite graph $G(V \cup U, A)$ with non-negative weights on the edges, and a matroid $M = (V, I)$, such that for every $T \subseteq U$, $f(T)$ is the total weight of the max weight matching on the subgraph induced by $T$ and a subset of $V$ that belongs to $I$.
\end{itemize}

The last four classes are part of the complement-free hierarchy \cite{lehmann2001combinatorial}, where containment is strict (i.e., 
$ Rado \subset GS \subset Submodular \subset XOS$).

\subsection{Structure of Optimal Contracts}
We begin by making a few observations on the structure of the utilities of both the agent and the principal. 
For a fixed set of actions $S \subseteq A$, the agent's utility, $u_a(\alpha, S) = \alpha f(S) - c(S)$, is an affine function of $\alpha$.
Thus, the utility of the agent 
$
u_a(\alpha) = \alpha f(S_\alpha) - c(S_\alpha)
$, 
is the upper envelope of these $2^n$ affine functions. As such, it is continuous, monotonically non-decreasing, convex and piece-wise linear, as illustrated in Figure~\ref{fig:ex1Util}.

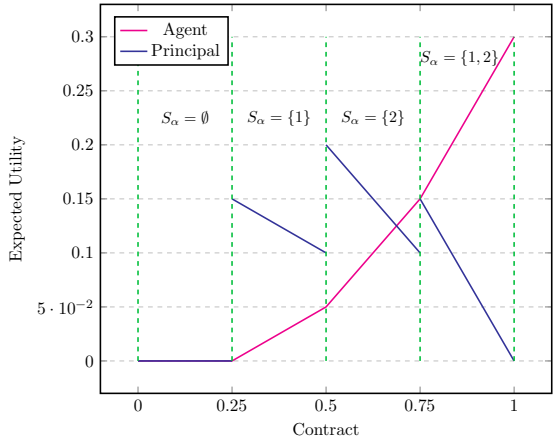
\begin{figure}[H]
\centering
\begin{tikzpicture}[scale=0.6]
\begin{axis} 
[   
    scale only axis,
    xlabel={Contract},
    ylabel={Expected Utility},
    xtick={0, 0.25, 0.5, 0.75, 1},
    ytick={0, 0.05, 0.1, 0.15, 0.2, 0.25, 0.3},
    legend pos=north west,
    ymajorgrids=true,
    grid style=dashed,
    line width=1pt
]
\addplot[color=magenta]
    coordinates{(0,0)(0.25,0)(0.5,0.05)(0.75,0.15)(1,0.3)};
\addplot[color=blue(pigment)] coordinates{(0,0)(0.25,0)};
\addplot[color=blue(pigment)] coordinates{(0.25,0.15)(0.5,0.1)};
\addplot[color=blue(pigment)] coordinates{(0.5,0.2)(0.75,0.1)};
\addplot[color=blue(pigment)] coordinates{(0.75,0.15)(1,0)};

\addplot[darkpastelgreen, dashed, fill opacity=1] coordinates{(0,0)(0,0.3)};
\addplot[color=darkpastelgreen, dashed, fill opacity=1] coordinates{(0.25,0)(0.25,0.3)};
\addplot[color=darkpastelgreen, dashed, fill opacity=1] coordinates{(0.5,0)(0.5,0.3)};
\addplot[color=darkpastelgreen, dashed, fill opacity=1] coordinates{(0.75,0)(0.75,0.3)};
\addplot[color=darkpastelgreen, dashed, fill opacity=1] coordinates{(1,0)(1,0.3)};

\node[] at (12.5,225) {\small $S_\alpha = \emptyset$};
\node[] at (37.5,225) {\small $S_\alpha = \{1\}$};
\node[] at (62.5,225) {\small $S_\alpha = \{2\}$};
\node[] at (85.5,280) {\small $S_\alpha = \{1,2\}$};

\legend{Agent, Principal}
\end{axis}
\end{tikzpicture}
\caption{\footnotesize The utilities of the principal and the agent as a function of $\alpha$ (critical values marked in green). $A = \{1,2\}$,
the rewards are $f(\emptyset)=0$,
$f(\{1\})=0.2$, $f(\{2\})=0.4$, $f(\{1,2\})=0.6$
and the costs $c(\emptyset)=0,$
$c(\{1\})=0.05$, $c(\{2\})=0.15$, $c(\{1,2\})=0.3$.}
\label{fig:ex1Util}
\end{figure}

The values of $\alpha$ at which the agent's demand changes are called \textit{critical values}; the set of critical values is denoted by 
$\critical_{f,c}$.

For every interval of $\alpha$ between two critical values, the agent's demand is fixed at some set $S$,
and the principal's utility is $u_p(\alpha) = (1-\alpha)f(S)$. 
Thus, the principal's utility function is linear-in-parts, and the best contract within each such interval is its left-most point (which is a critical value).
Indeed, within the interval, the agent's demand does not change, and $\alpha$ increases, thus the principal's utility $u_p(\alpha) = (1-\alpha)f(S)$ decreases.
This immediately implies that the optimal contract must be a critical value.
This is cast in the following observation.

\begin{observation} [Dütting et al.~\cite{dutting2022combinatorial}]\label{obs:critVals}
    For every cost and reward functions 
    $c: 2^A \to \reals_{\ge 0}$ and $f: 2^A \to [0,1]$, there exist $k \le 2^n + 1$ critical values, $0=\alpha_0 < \alpha_1 < \hdots < \alpha_k \le 1$, such that 
    \begin{enumerate}
        \item For every contract $x \in [0,1]$, $S_x = S_{\alpha_{i(x)}}$, where $i(x) = \argmax_{j}\{\alpha_j \mid \alpha_j \le x\}$.
        \item For every $0 \le i < j \le k$, $f(S_{\alpha_i}) < f(S_{\alpha_j})$. 
        Furthermore, $c(S_{\alpha_i}) < c(S_{\alpha_j})$ under consistent tie-breaking.
        \item The optimal contract, $\alpha^* = \argmax_{\alpha \in [0,1]} u_p(\alpha^*)$, satisfies $\alpha^* \in \{\alpha_0,\alpha_1,...,\alpha_k\} = \critical_{f,c}$. 
    \end{enumerate}
\end{observation}

\Cref{obs:critVals} simplifies the optimal contract problem as it allows us to restrict attention to the set of critical values. 
And yet, even with this restriction, the problem is far from being trivial, due to two main challenges. 
First, an instance may admit exponentially many critical values. 
Second, unlike the agent's utility, the principal's utility is not well-behaved. For example, it may be neither monotone nor continuous in $\alpha$. 

Note also that the cost and reward functions, $c$ and $f$, are set functions of exponential size (in the number of actions $n$). 
We consider two oracle models to these functions, named value and demand oracles. 
A value oracle for $f$ (respectively, for $c$) receives a set $S\subseteq A$ and returns $f(S)$ (resp., $c(S)$). 
A demand oracle for $f,c$ receives a contract $\alpha \in [0,1]$ and returns the agent's best response, $S_\alpha$.
Note that this corresponds to a demand query with bundle prices, where the value function is $f(S)$ and the price of a bundle $S$ is $c(S)/\alpha$. In the case where $c$ is additive we obtain the canonical definition of a demand query for $f$ with item prices.

\subsection{A General Scheme for the Optimal Contract}
Following the above observations, the authors in \cite{dutting2022combinatorial} suggest a scheme for computing the optimal contract (Algorithm 4.1 in their paper): Going over all critical values of $\alpha$ and pick the one that maximizes the principal's utility. This algorithm runs in polynomial time if the following three conditions hold:
\begin{enumerate}
    \item \label{prop1} The number of critical values, $|\critical_{f,c}|$, is polynomial in $n$.
    \item \label{prop2} For every $\alpha \in [0,1]$ there is an efficient algorithm for computing the agent's demand.
    \item \label{prop3} For every $\alpha \in [0,1]$ there is an efficient algorithm for finding the next critical value.
\end{enumerate}

This paradigm is used in \cite{dutting2022combinatorial} to find the optimal contract when rewards are gross substitutes and costs are additive. For establishing property \ref{prop3}, the authors of \cite{dutting2022combinatorial} rely on the greedy implementation of the agent's demand oracle \cite{Bertelsen05,PaesLeme17}.
For the larger class of submodular rewards, they exploit that the demand oracle problem is NP-hard (i.e., property \ref{prop2}.~is violated), to show that the problem of finding the optimal linear contract is NP-hard. They also give a submodular function for which there are exponentially many critical values (i.e., property \ref{prop1}.~is violated).

\section{Enumerating all Critical Values}
In this section we present an efficient algorithm for finding all critical values. That is, given oracle access to the agent's demand, the algorithm returns all values of $\alpha$ at which the agent's demand changes. The running time of the algorithm is linear in the number of critical values.
The algorithm doesn't assume anything regarding the structure of the agent's demand problem, nor the implementation of its oracle.

This algorithm has been discovered several times in different contexts. In particular, it has been used by Gusfield \cite{gusfield1980sensitivity} in the context of sensitivity analysis, where it is referred to as the Eisner-Severance technique \cite{eisner1976mathematical}. Here the problem is, given a parameterized combinatorial optimization problem, find all the breakpoints where the optimal solution changes. These breakpoints correspond to the critical values in our contract problem.

We thus show:

\begin{proposition}\label{prop:optContractSuff}
    Let $c:2^A \to \reals_{\ge 0}$ and $f:2^A \to [0,1]$ be two functions that describe the cost and reward of a given set of actions, such that $|\critical_{f,c}|$ is polynomial in $|A|=n$. Assume further that there exists an efficient algorithm to compute the agent's demand, then the optimal contract can be computed efficiently.
\end{proposition} 

Note how this theorem reduces the requirements for the general scheme of \cite{dutting2022combinatorial} from three to two conditions by showing that properties \ref{prop1}.~and \ref{prop2}.~implies property \ref{prop3}.~Before proving \cref{prop:optContractSuff}, we describe the algorithm for enumerating all critical values and prove its correctness.

Recall the agent's utility function $u_a: [0,1] \to \reals$, given its best response to the contract $\alpha$, where we break ties towards the set with the higher reward and in some consistent manner:
$$
u_a(\alpha) = 
\max_{S\subseteq A} \alpha f(S) - c(S) = 
\max_{S\subseteq A} u(\alpha,S)
$$
It is clear that $u_a$ is piece-wise linear as a function of $\alpha$. 
Moreover, as the maximum over linear functions, it also exhibits convexity, as established in the following lemma
\begin{lemma}\label{lem:utilityConvex}
    Let $S_\alpha, S_\beta \subseteq A$ be two different sets 
    that maximize the agent's utility for two  different contracts $0 \le \alpha < \beta \le 1$.
    Then,
    \begin{enumerate}
        \item $f(S_\alpha) < f(S_\beta)$
        \item $c(S_\alpha) < c(S_\beta)$
    \end{enumerate}
\end{lemma}
\begin{proof}
    The fact that $S_\alpha$ is superior for contract $\alpha$ and $S_\beta$ is superior for $\beta$ implies that
    \begin{eqnarray*}
        \alpha f(S_\beta) - c(S_\beta) \le \alpha f(S_\alpha) - c(S_\alpha) \\
        \beta f(S_\beta) - c(S_\beta) \ge \beta f(S_\alpha) - c(S_\alpha)
    \end{eqnarray*}
    Rearranging,
    $$
    \alpha [f(S_\beta) - f(S_\alpha)] \le \beta [f(S_\beta) - f(S_\alpha)]
    $$
    Together with the fact that $\alpha < \beta$ we get that $f(S_\alpha) \le f(S_\beta)$.
    Due to consistent tie breaking, it must be  that $f(S_\alpha) < f(S_\beta)$.
    For the second part of the claim, observe that if $c(S_\alpha) \ge c(S_\beta)$, then
    $$
    u(\alpha,S_\beta) = \alpha f(S_\beta) - c(S_\beta) \ge \alpha f(S_\alpha) - c(S_\alpha) = u(\alpha,S_\alpha)
    $$
    Which contradicts the optimality of $S_\alpha$.
\end{proof}

Our algorithm for enumerating the set of critical values for a given segment, $[\alpha, \beta]$, makes use of the utility's piece-wise linear structure to query the agent's demand oracle with contracts that may be critical.
Because critical values are contracts for which the agent's demand change, those are points of indifference between two sets of actions. Namely, for any critical point $\gamma$ there exists at least two sets of actions $S, T \subseteq A$ with $u_a(\gamma,S) = u_a(\gamma, T)$.
This observation, together with the linear structure, allow us to identify all critical points using a recursive strategy:

Starting with the end-points of the segment $[\alpha, \beta]$, we query the oracle. If $S_\alpha = S_\beta$, there are no critical values in this segment. If the sets $S_\alpha$ and $S_\beta$ differ, then we have a critical value in this segment. If there is a single critical value, it can only be the point of indifference between those two sets: $\gamma = \frac{c(S_\beta) - c(S_\alpha)}{f(S_\beta) - f(S_\alpha)}$ and the oracle will return $S_\gamma = S_\beta$ (by tie breaking consistency).
If there is more than one critical value, we will get $S_\gamma \ne S_\beta$, and the critical values can be found by solving for both sub-segments $[\alpha,\gamma]$ and $[\gamma,\beta]$. The full description of the algorithm is presented below.

\begin{algorithm}[H]
\caption{Find All Critical Values in a Segment. \texttt{CV}($\alpha$, $\beta$)}\label{alg:CV}
\begin{algorithmic}[1]
\If{$S_\alpha = S_\beta$}
    \State return $\emptyset$
\EndIf
\State $\gamma \gets \frac{c(S_\alpha) - c(S_\beta)}{f(S_\alpha) - f(S_\beta)}$
\Comment{$\gamma$ is the intersection between $u_a(\cdot, S_\alpha)$ and $u_a(\cdot, S_\beta)$}
\If{$S_\gamma = S_\beta$}
    \State return $\{\gamma\}$
\EndIf
\State $C_1 \gets \texttt{CV}(\alpha, \gamma)$
\State $C_2 \gets \texttt{CV}( \gamma, \beta)$
\State return $C_1 \cup C_2$
\end{algorithmic}
\end{algorithm}

\begin{claim}\label{cla:CVcorrect}
    For any segment $[\alpha, \beta]$, $\texttt{CV}(\alpha, \beta)$ returns all the critical values in that segment.
\end{claim}
\begin{proof}
    We prove the claim by induction over the number of critical values in the segment.
    
    For the basis of the induction, assume there are no critical values in $[\alpha, \beta]$. Thus, it must be that the agent's utility is linear in that segment, by the consistency of tie breaking $S_\alpha = S_\beta$ ---- and the algorithm returns $\emptyset$.
    
    If there exists a single critical value, since the agent's utility is piece-wise linear, $S_\alpha \ne S_\beta$ and the critical value is at the intersection of the linear functions $u_a(x, S_\alpha)= xf(S_\alpha) - c(S_\alpha)$ and $u_a(x, S_\beta)= xf(S_\beta) - c(S_\beta)$, which is exactly $\gamma = \frac{c(S_\beta) - c(S_\alpha)}{f(S_\beta) - f(S_\alpha)}$.    
    Note that this implies that $S_\gamma = S_\beta$, as ties are broken in favor of the principal.

    Assume the correctness of the algorithm for any a segment with $k > i \ge 1$ critical values and consider the segment $[\alpha, \beta]$ with $k$ critical values.
    The inductive step will show that there is at least one critical value on each side of $\gamma$, so each of the sub-segments $[\alpha, \gamma]$ and $[\gamma, \beta]$ have strictly less than $k$ critical values.
    Thus, when calling $\texttt{CV}(\alpha,\gamma)$ and $\texttt{CV}(\gamma,\beta)$, by the induction hypothesis, the critical values for each of these sub-segments are returned. As implied by the base of the induction, if $\gamma$ is a critical value it will emerge from the $\texttt{CV}(\alpha,\gamma)$ branch in the calls tree.
    
    Let $\dmin$ be the smallest critical value in $(\alpha, \beta]$.
    As $S_\dmin$ is the indifference point between $S_\alpha$ and $S_\dmin$, $\dmin = \frac{c(S_\dmin) - c(S_\alpha)}{f(S_\dmin) - f(S_\alpha)}$.
    Aiming for contradiction, assume $\dmin \in [\gamma, \beta]$. Thus, $S_\alpha$ dominates $S_\dmin$ in the segment $[\alpha, \gamma)$ and at $\gamma$:
    $$    u_a(S_\beta,\gamma)=u_a(S_\alpha,\gamma) \ge u_a(S_\dmin, \gamma)
    $$
    By Lemma \ref{lem:utilityConvex}, $f(S_\beta)>f(S_\dmin)$ and so $S_\beta$ dominates $S_\dmin$ in $[\gamma,\beta]$, contradicting the optimality of $S_\dmin$.

    An analogous argument can be applied to show the existence of a critical value in the segment $(\gamma, \beta]$: 
    let $\dmax$ be the largest critical value in $[\alpha, \beta]$. Observe that $S_\dmax=S_\beta$, and let $S$ be the optimal set just preceding $S_\beta$ (that is, for every sufficiently small $\eps>0$, $S_{\dmax-\eps} = S$).
    Aiming for contradiction, assume $\dmax \in [\alpha, \gamma]$. 
    As $\dmax$ is the indifference point between $S$ and $S_\beta$, and $f(S_\beta) > f(S)$, it holds that $S_\beta$ dominates $S$ in the segment $[\gamma,\beta]$.
    Under contract $\gamma$:
    $$
    u_a(S_\alpha, \gamma) = u_a(S_\beta, \gamma) \ge u_a(S, \gamma)
    $$
    Combined with $f(S)>f(S_\alpha)$ (by Lemma \ref{lem:utilityConvex}), we get that for any $\delta < \gamma$:
    \begin{eqnarray*}
    u_a(S_\alpha, \delta) &=& 
    u_a(S_\alpha, \gamma) - (\gamma-\delta)f(S_\alpha) \\
    &\ge&
    u_a(S, \gamma) - (\gamma-\delta)f(S_\alpha) \\
    &>&
    u_a(S, \gamma) - (\gamma-\delta)f(S) \\
    &=&
    u_a(S, \delta)    
    \end{eqnarray*}
    This implies that $S$ in never optimal, a contradiction.
\end{proof}

\begin{claim}\label{cla:CVruntime}
    The above algorithm takes linear time in the number of critical values in $[\alpha, \beta]$.
\end{claim}
\begin{proof}
    This can also be shown using an inductive argument.
    If the segment $[\alpha, \beta]$ has $n > 1$ critical values in it, by the description of the algorithm the recurrence relation describing the time complexity is of the form 
    $$
    T(n) = c_0 + T(n-k) + T(k)
    $$
    For some $1\le k<n$, where $c_0$ is the number of operation made in lines $1-7$ of \Cref{alg:CV}.
    The relation is true because each critical value is found in exactly one of the two sides of $\gamma$\footnote{if $\gamma$ is a critical point, it is processed in $\texttt{CV}(\alpha,\gamma)$, as discussed in the proof for correctness.}, and there is at least one critical value on each side. 
    
    For the basis of the induction, observe that $T(0), T(1) \le c_0$. 
    Thus, if we assume for any $l < n$ that $T(l) \le c_0(2l - 1)$, then by the above recurrence relation
    $$
    T(n) = c_0 + T(n-k) + T(k) \le c_0 + c_0(2(n-k) - 1) + c_0(2k - 1) = c_0(2n - 1),
    $$
    as claimed.
\end{proof}

\begin{proof}[Proof of ~\cref{prop:optContractSuff}]By claims \ref{cla:CVcorrect} and \ref{cla:CVruntime}, given an oracle access to the agent's demand and a polynomial number of critical values, we can efficiently find all critical values and by observation \ref{obs:critVals} one if these values constitutes the optimal contract.
\end{proof}

The sensitivity literature has also studied the number of breakpoints of several parameterized combinatorial optimization problems, typically with \emph{affine} parametrization \cite[e.g.,][]{carstensen1983complexity,gusfield1980sensitivity, karp1981parametric}. This is closely related to determining the number of critical values of $\alpha$ in linear contracts for combinatorial contracting problems. 
Our result (\Cref{prop:optContractSuff}) gives an immediate algorithmic implication of a  small (polynomial) number of critical values, namely a polynomial-time algorithm for finding the optimal linear contract. 

\section{Submodular Costs and Supermodular Rewards}
In this section, we use \Cref{prop:optContractSuff} to show that the optimal contract problem for complementary actions can be found in polynomial time. We model complementary actions using monotone submodular cost functions and monotone supermodular reward functions.
In this case, the marginal cost of an action is non-increasing in the set of actions already taken and the marginal reward is non-decreasing.

\begin{theorem}\label{thm:submodular-costs}
    Let $c:2^A \to \reals$ be a monotone submodular function and let $f: 2^A \to [0,1]$ be a monotone supermodular function, then there exists an efficient algorithm for computing the optimal contract. 
\end{theorem}

We will use \Cref{prop:optContractSuff} and show the two sufficient conditions specified to guarantee the existence of an efficient algorithm for this problem: polynomially-many critical values and an efficient algorithm for the agent's demand.

\paragraph{Tractability of the Agent's Demand Problem.}
Where the agent's utility function, $u(\alpha, S) = \alpha f(S) - c(S)$, 
is such that $c$ is submodular, and $f$ is supermodular it can be easily seen that for any fixed $\alpha$, $u(\alpha, \cdot)$ is supermodular. 
As a demand oracle for the agent requires finding the set $S_\alpha \subseteq A$ which maximizes $u(\alpha, \cdot)$, this is equivalent to finding $S_\alpha$ which minimizes the submodular function $-u(\alpha, \cdot)$.
It is well known that there exists a polynomial time algorithm for submodular function minimization, commonly known as SFM  (e.g. \cite{iwata2001combinatorial}).
Thus, the agent's problem can be solved in polynomial time and a demand oracle for the agent can be implemented efficiently.

\begin{claim}\label{cla:submod-costs-agent-utility}
There exists an efficient algorithm that computes the demand of the agent for any contract $\alpha$.
\end{claim}

\paragraph{Linear Number of Critical Values.}
The supermodularity of the agent's utility does not only imply the existence of an efficient demand oracle, but also that there are no more than $n$ critical values. We show that in fact, the set of actions which is the agent's best response may only be extended as we increase the contract $\alpha$. This immediately implies that the number of optimal sets cannot surpass $n$.

\begin{claim}\label{cla:SuperModContainment}
    Let $c$ be a monotone submodular cost function and $f$ a monotone supermodular reward function, then for any two contracts $\alpha < \alpha'$ and two corresponding sets in the agent's demand $S_\alpha$, $S_{\alpha'}$ it holds that $S_\alpha \subseteq S_{\alpha'}$.
\end{claim}

\begin{proof}
    If $S_\alpha = S_{\alpha'}$ the claim obviously hold.
    Otherwise, assume that $S_{\alpha'}$ is a maximal best-response for contract $\alpha'$ (this is in line with our tie-breaking assumption), and also that 
    $S_\alpha \setminus S_{\alpha'} = R$ is such that $R \ne \emptyset$, we will show that a contradiction is reached.
    By the fact that $S_\alpha$ is optimal for $\alpha$, it must be that 
    $$
    u_a(\alpha, R \mid S_\alpha \cap S_{\alpha'}) = u_a(\alpha, S_\alpha) - u_a(\alpha, S_\alpha \cap S_{\alpha'}) \ge 0
    $$
    By the supermodularity of $f$ and submodularity of $c$ it holds that $f(R \mid S_\alpha \cap S_{\alpha'}) \le f(R \mid S_{\alpha'})$ and $c(R \mid S_\alpha \cap S_{\alpha'}) \ge c(R \mid S_{\alpha'})$. Putting everything together we get
    \begin{eqnarray*}
        u(\alpha', R \mid S_{\alpha'}) 
        &=&
        \alpha' f(R \mid S_{\alpha'}) - c(R \mid S_{\alpha'}) \\
        &\ge&
        \alpha' f(R \mid S_\alpha \cap S_{\alpha'}) - c(R \mid S_\alpha \cap S_{\alpha'}) \\ 
        &\ge&
        \alpha f(R \mid S_\alpha \cap S_{\alpha'}) - c(R \mid S_\alpha \cap S_{\alpha'}) \\ 
        &=&
        u(\alpha, R \mid S_\alpha) \\
        &\ge&
        0
    \end{eqnarray*}
     Where the second inequality follows from the monotonicity of $f$, which imply $f(R \mid S_\alpha \cap S_{\alpha'}) \ge 0$.
     Thus, we can add $R$ to $S_{\alpha'}$ while not losing utility, contradicting its maximality.
\end{proof}

It is worth mentioning that it is possible to prove Theorem~\ref{thm:submodular-costs} via the blueprint in \cite{dutting2022combinatorial}, which requires an efficient algorithm for solving the agent's demand problem, a polynomial number of critical values, and a polynomial time procedure to compute the next critical value.

The existence of the first two conditions is established in Claim \ref{cla:submod-costs-agent-utility} and Claim \ref{cla:SuperModContainment} above. The task of finding the next critical value can be solved by adopting ideas from Nagano \cite{nagano2007strongly}. Nagano's goal is to solve the minimum ratio problem, a natural question when considering parametric submodular minimization, and they accomplish this using the parametric search technique of Meggido \cite{megiddo1978combinatorial}.

Using Nagano's approach, finding the next critical value requires following a submodular function minimization algorithm which is fully-combinatorial (i.e. only uses additions, subtractions and comparisons), and for every comparison it makes run another instance of any SFM algorithm.
The best known fully combinatorial SFM algorithm runs in $\tilde{O}(n^8)$ steps \cite{iwata2002fully}.
Our approach is both simpler and more 
efficient, since we compute all critical values by calling $O(n)$ times to a demand oracle, which can be implemented using any SFM algorithm (e.g. \cite{iwata2003faster}, which runs in $\tilde{O}(n^7)$).

\section{Matching-Based Rewards and Additive Costs}

In this section we consider a settings in which an agent matches between tasks and resources.
This is captured by a bipartite graph $G=(V \cup U, A)$ where $V$ is the set of tasks and $U$ is the set of resources. 
An edge $(v,u) \in A$ indicates that resource $u$ can be used to complete task $v$. 
Each such edge is associated with cost $c_{v,u} \ge 0$ and reward $f_{v,u} \ge 0$.
For a set $S \subseteq A$ of edges (actions) chosen by the agent, its cost $c(S)$ is additive and its reward $f(S)$ is the max weight matching in the subgraph induced by $S$, with respect to weights $\{f_{v,u}\}_{(v,u)\in A}$.
We call such a function $f$ 
\emph{matching-based reward function}. 

It is easy see that every matching based reward function is XOS, a superset of submodular functions. \Cref{ex:MatchingNotSubNotSup} shows that some matching-based functions are strictly XOS.

\begin{example}\label{ex:MatchingNotSubNotSup}
    Consider the bipartite graph with vertices $V = \{v_1,v_2\}$ and $U = \{u_1,u_2\}$ and edges $A = \{(v_1,u_1), (v_1,u_2),(v_2,u_1)\}$, with rewards (see \Cref{fig:nonSubSupMod}):
    $$f(v_1,u_1) = f(v_1,u_2) = 2  \qquad f(v_2,u_1) = 1$$
    Observe that 
    $f(v_2,u_1) = 1$
    and 
    $f((v_2,u_1) \mid \{(v_1,u_1)\}) = 0$,
    which suggests $f$ is not supermodular.
    However, since the edges $\{(v_1,u_2), (v_2,u_1)\}$ make up a more rewarding matching than $\{(v_1,u_1)\}$, we have that 
    $f((v_2,u_1) \mid \{(v_1,u_1), (v_1,u_2)\}) = 1$, so $f$ is not submodular either. 
\end{example}

\begin{figure}[H]
\centering    
\begin{tikzpicture}[scale=0.125]
\tikzstyle{every node}+=[inner sep=0pt]
\draw [black] (0,0) circle (3); \draw (0,0) node {$v_1$};
\draw [black] (0,-15) circle (3); \draw (0,-15) node {$v_2$};
\draw [black] (30,0) circle (3); \draw (30,0) node {$u_1$};
\draw [black] (30,-15) circle (3); \draw (30,-15) node {$u_2$};
\draw [black] (3,0) -- (27,0); \draw (15,0) node [above] {\small$f=2$};
\draw [black] (3,0) -- (27,-15); \draw (21,-11) node [rotate=-32, above] {\small$f=2$};
\draw [black] (3,-15) -- (27,0); \draw (9,-11.5) node [rotate=32, below] {\small$f=1$};
\end{tikzpicture}
\caption{Matching-based instance with $f$ not submodular nor supermodular}
\label{fig:nonSubSupMod}
\end{figure}
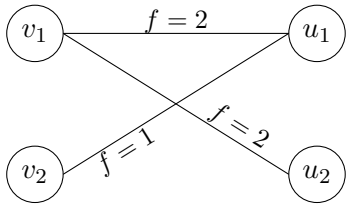

In this setting, the agent's utility for a set of edges $S$ and a contract $\alpha$ is given by
$$
u_a(\alpha, S) = \max_{\substack{S' \subseteq S, \\ S' \text{ is a matching}}} \sum_{ (v,u) \in S'} \alpha f_{v,u}  - \sum_{(v,u) \in S} c_{v,u}
$$
and its best response for a contract $\alpha$ is without loss of generality a set $S_\alpha \subseteq A$ which constitutes a matching.
Thus, the agent's demand problem corresponds to finding a maximum weight matching where the weight of an edge $(v,u)$ is $w_\alpha(v,u) = \alpha f_{v,u} - c_{v,u}$ (note that the weights may be negative). 
As such, the demand can be computed efficiently using any algorithm for maximum weighted matching, so the first condition mentioned in  \Cref{prop:optContractSuff} is met.

In \Cref{subsec:superpoly} we show that the second condition of \Cref{prop:optContractSuff} is not always fulfilled, and present family of matching-based instances for which the number of critical values is super-polynomial.
In \Cref{subsec:PolyMatchingBased}, we give two special cases for matching-based rewards and costs in which there are polynomially-many critical values and the optimal contract is tractable. 

\subsection{A Super-Polynomial Lower Bound on the Number of Critical Values}
\label{subsec:superpoly}

We derive a super-polynomial lower bound on the number of critical values from the construction made by Mulmuley and Shah \cite{mulmuley2000lower}, based on the work of Carstensen \cite{carstensen1983complexity}, 
for the parameterized shortest $s-t$ path problem.
In this problem, we are given a directed graph with two special nodes, $D=(X \cup \{s,t\}, \Tilde{E})$, and a set of affine weight functions, $\{\Tilde{w}_{\alpha}(e) = \Tilde{c}_e - \alpha \Tilde{f}_e\}_{e\in E}$, 
of a parameter $\alpha \in \reals_{\ge 0}$.  
The weights are such that for every $e \in E$, $\Tilde{c}_e \ge 0$. 
The goal is to find all the shortest paths between $s$ and $t$ for a given range of values of $\alpha$.
\begin{theorem}[Theorem 1.3, \cite{mulmuley2000lower}]\label{thm:MulShah}
    There is a family of graphs $D_n=(X_n \cup \{s,t\}, \Tilde{E}_n)$ with $O(n^4)$ vertices and affine weight functions $\{\Tilde{w}_\alpha(e) = \Tilde{c}_e - \alpha \Tilde{f}_e \mid e \in \Tilde{E}_n\}$, together with an upper bound $Q_n \in \reals_{\ge 0}$, such that
    \begin{enumerate}
        \item For every $e \in \Tilde{E}_n$, $\Tilde{c}_e \ge 0$
        \item For every $\alpha \in [0,Q_n]$ and for every $e\in \Tilde{E}_n$, $\Tilde{w}_e = \Tilde{c}_e - \alpha \Tilde{f}_e \ge 0$
        \item There are $m = 2^{\Omega(\log^2n)}$ values, $0 \le \alpha_1 <...<\alpha_m < Q_n$ such that for any two $\alpha_j \ne \alpha_i$, the shortest $s-t$ path with respect to weights $\Tilde{w}_{\alpha_i}$ is different from the shortest $s-t$ path with respect to $\Tilde{w}_{\alpha_j}$.
    \end{enumerate}
\end{theorem}

In order to apply the above lower bound to the number of critical values, we present a reduction from the parametric shortest $s-t$ path to the agent's demand problem for matching-based costs and rewards.
We begin by reducing an instance of the shortest $s-t$ path problem for a fixed value of $\alpha$ to the minimum weight perfect matching problem in a bipartite graph. We will show that the shortest $s-t$ has a corresponding minimum perfect weight matching of the same weight. This reduction will preserve the parametric structure of the weight functions, so we will get $m$ different perfect weight matchings for $\alpha_1,...,\alpha_m$.

This will be followed by another reduction from the minimum perfect weight matching to the agent's demand problem. The affine structure of the weights will remain, but we will manipulate the coefficients in a way that will make this maximum weight matching problem a valid instance of the agent's demand problem.

\paragraph{Shortest $s-t$ Path to Minimum Perfect Matching.}

Given an instance of the $s-t$ path problem $(D,\Tilde{w}_\alpha)$, where $D=(X\cup\{s,t\}, \Tilde{E})$ is a directed graph with non-negative weights $\{\Tilde{w}_e = \Tilde{c}_e - \alpha\Tilde{f}_e\}_{e \in \Tilde{E}}$ on its edges, we construct an instance of the minimum perfect weight matching problem, $(G',w'_{\alpha})$. In this problem, we need to find the minimum perfect matching in the undirected bipartite graph $G'=(V'\cup U', E')$ with respect to weights defined over the edges $w'_{\alpha}:E' \to \reals$.
The instance we construct is defined as follows (see illustration in Figure \ref{fig:stReduction}).

\begin{figure}[H]
\centering    
\begin{tikzpicture}[scale=0.165]
\tikzstyle{every node}+=[inner sep=0pt]
\draw [black] (29.7,-20.4) circle (3);
\draw (29.7,-20.4) node {$t$};
\draw [black] (19.3,-10.9) circle (3);
\draw (19.3,-10.9) node {$x_1$};
\draw [black] (8.8,-20.4) circle (3);
\draw (8.8,-20.4) node {$s$};
\draw [black] (19.3,-29.4) circle (3);
\draw (19.3,-29.4) node {$x_2$};
\draw [black] (11.02,-18.39) -- (17.08,-12.91);
\fill [black] (17.08,-12.91) -- (16.15,-13.08) -- (16.82,-13.82);
\draw (11.17,-15.16) node [above] {$w_{s,1}$};
\draw [black] (18.748,-26.452) arc (-171.42741:-188.57259:42.277);
\fill [black] (18.75,-26.45) -- (19.12,-25.59) -- (18.13,-25.74);
\draw (17.78,-20.15) node [left] {$w_{1,2}$};
\draw [black] (17.02,-27.45) -- (11.08,-22.35);
\fill [black] (11.08,-22.35) -- (11.36,-23.25) -- (12.01,-22.49);
\draw (11.76,-25.39) node [below] {$w_{2,s}$};
\draw [black] (21.57,-27.44) -- (27.43,-22.36);
\fill [black] (27.43,-22.36) -- (26.5,-22.51) -- (27.15,-23.26);
\draw (26.68,-25.39) node [below] {$w_{2,t}$};
\draw [black] (21.51,-12.92) -- (27.49,-18.38);
\fill [black] (27.49,-18.38) -- (27.23,-17.47) -- (26.56,-18.21);
\draw (26.69,-15.16) node [above] {$w_{1,t}$};
\draw [black] (19.839,-13.851) arc (8.37112:-8.37112:43.27);
\fill [black] (19.84,-13.85) -- (19.46,-14.71) -- (20.45,-14.57);
\draw (20.8,-20.15) node [right] {$w_{2,1}$};
\draw [black] (45.2,-8.3) circle (3);
\draw (45.2,-8.3) node {$s$};
\draw [black] (45.2,-17.1) circle (3);
\draw (45.2,-17.1) node {$v_1$};
\draw [black] (45.2,-26.3) circle (3);
\draw (45.2,-26.3) node {$v_2$};
\draw [black] (73,-17.1) circle (3);
\draw (73,-17.1) node {$u_1$};
\draw [black] (73,-26.3) circle (3);
\draw (73,-26.3) node {$u_2$};
\draw [black] (73,-35.9) circle (3);
\draw (73,-35.9) node {$t$};
\draw [black] (48.2,-26.3) -- (70,-26.3);
\draw (64,-25.8) node [above] {$0$};
\draw [black] (48.2,-17.1) -- (70,-17.1);
\draw (64,-16.6) node [above] {$0$};
\draw [black] (48.05,-25.36) -- (70.15,-18.04);
\draw (66,-20) node [below] {$w_{2,1}$};
\draw [black] (48.06,-9.21) -- (70.14,-16.19);
\draw (59,-12.3) node [above] {$w_{s,1}$};
\draw [black] (48.05,-18.04) -- (70.15,-25.36);
\draw (56,-20) node [above] {$w_{1,2}$};
\draw [black] (47.69,-18.78) -- (70.51,-34.22);
\draw (59.8,-28) node [below] {$w_{1,t}$};
\draw [black] (48.04,-27.28) -- (70.16,-34.92);
\draw (59,-31.67) node [below] {$w_{2,t}$};
\end{tikzpicture}
\caption{A reduction from a directed graph for the $s-t$ problem (left) to a bipartite graph for the minimum perfect weight matching (right).}
\label{fig:stReduction}
\end{figure}
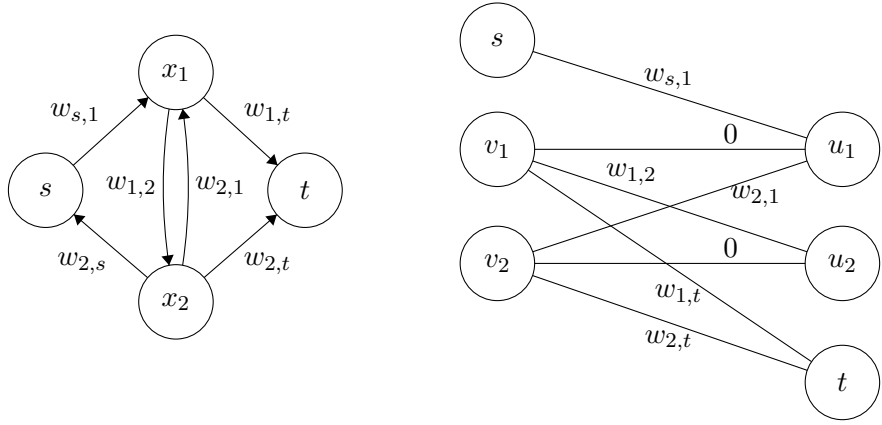

Each vertex $x_i \in X$ is associated with two vertices in $G'$, $v_i \in V'$ and $u_i \in U'$ and the vertices $s$ and $t$ in $D$ correspond to $s \in V'$ and $t \in U'$ respectively.
The edges, $E'$, are defined as follows:
For every directed edge $(x_i, x_j) \in \Tilde{E}$, where $x_i, x_j \in X$, we have a corresponding $(v_i, u_j) \in E'$. For every edge that leaves $s$ in the graph $D$, we have a corresponding one in $E'$. That is, $(s,x_i)\in \Tilde{E} \to (s,u_i) \in E'$ and $(s,t)\in \Tilde{E} \to (s,t) \in E'$. We also define corresponding edges to those that enter $t$, namely $(x_i,t)\in \Tilde{E} \to (v_i,t) \in E'$ (We ignore any edge that enters $s$ or leaves $t$, as those are never part of the optimal $s-t$ path). We also add an edge $(v_i,u_i) \in E'$ for every $i \in \{1,...,|X|\}$.
For each $(v,u) \in E'$ we define the weights $w'$ as follows:
$$
w'_\alpha(v,u) = 
\begin{cases}
    \Tilde{w}_\alpha(x_i,x_j) = \Tilde{c}_{i,j} - \alpha\Tilde{f}_{i,j} & v=v_i, u=u_j, j\ne i\\
    \Tilde{w}_\alpha(s,x_i) = \Tilde{c}_{s,i} - \alpha\Tilde{f}_{s,i} & v=s, u=x_i\\
    \Tilde{w}_\alpha(x_i,t) = \Tilde{c}_{i,t} - \alpha\Tilde{f}_{i,t} & v=v_i, u=t\\
    0 = 0 - \alpha 0 & v=v_i, u=u_i
\end{cases}
$$
We claim that this construction is such that for any $\alpha$, the shortest $s-t$ path in $(D,\Tilde{w}_\alpha)$ has a corresponding minimum perfect matching in $(G,w'_\alpha)$ with the same weight. 

\begin{claim}\label{cla:minPerfectReduction}
    For any value of $\alpha$, the shortest $s-t$ path in $(D,\Tilde{w}_\alpha)$, $P_\alpha$ has a corresponding minimum perfect weight matching in $(G',w'_\alpha)$, $M_\alpha$ such that both have the same weight.
    
\end{claim}
\begin{proof}
    Fix $\alpha$ and assume without loss of generality that $P_\alpha=(s,x_1,...,x_k,t)$ is the shortest $s-t$ path in $D$, then by definition of $G'$ there exists a perfect matching of the same weight,
    $
    M=\{(s,u_1),(v_1,u_2),...,(v_{k-1},u_k),(v_k,t)\}\cup\{(v_j,u_j) \mid j \notin [k]\}
    $
    .

    Let $M'_\alpha$ be the minimum weight perfect matching in $(G',w'_\alpha)$ and assume that its weight is strictly lesser than the weight of $M$ above.
    We will construct an $s-t$ path, $P$, in $(D,\Tilde{w}_\alpha)$, which is shorter than $P_\alpha$ and reach contradiction.
    Let us rename the vertices such that $(s,u'_1),(v'_1,u'_2),\ldots,$ $(v'_{l-1},u'_l),(v'_l,t) \in M'_\alpha$. 
    We can assume the rest of the edges in $M'_\alpha$ are of the form $(v'_j,u'_j)$, as those are edges with minimal weights. 
    Observe that the path $P=(s,x'_1,...,x'_l,t)$ is an $s-t$ path in $D$ with the same weight as $M_\alpha$. Thus, by assumption
    $$
        \Tilde{w}_\alpha(P) = w'_\alpha(M'_\alpha) < w'_\alpha(M) = \Tilde{w}_\alpha(P_\alpha)
    $$
\end{proof}

\Cref{cla:minPerfectReduction} above together with  \Cref{thm:MulShah}, imply that there exists a family of bipartite graphs with affine weight functions such that the set of minimum perfect weight matchings, for all values of the parameter $\alpha$, corresponds exactly to the set of all shortest $s-t$ paths.

\begin{corollary}\label{cor:perfectMinMatching}
    There exists a family of bipartite graphs $G'_n=(V'_n \cup U'_n, E'_n)$ with $O(n^4)$ vertices and affine weight functions $\{w'_\alpha(e) = c'_e - \alpha f'_e \mid e\in E'_n\}$, together with an upper bound $Q_n \in \reals_{\ge 0}$, such that
    \begin{enumerate}
        \item For every $e \in E'_n$, $c'_e \ge 0$
        \item For every $\alpha \in [0,Q_n]$ and for every $e\in E_n$, $w'_e = c'_e - \alpha f'_e \ge 0$
        \item There are $m = 2^{\Omega(\log^2n)}$ values, $0 \le \alpha_1 <...<\alpha_m < Q_n$ such that for any two $\alpha_j \ne \alpha_i$, the minimum weight perfect matching with respect to weights $w'_{\alpha_i}$ is different from the minimum weight perfect matching with respect to $w'_{\alpha_j}$.
    \end{enumerate}
\end{corollary}

\paragraph{Minimum Weight Perfect Matching to the Agent's Demand Problem.}
Given a minimum weight perfect matching instance $(G',w')$ as above, we construct a bipartite graph $G=(V\cup U, E)$ with weights $w:E \to \reals_{\ge 0}$ that correspond to the agent's demand problem. The vertices and edges of $G$ are the same as $G'$, so $V=V'$, $U=U'$ and $E=E'$, we only modify the weights.

Observe that if we just negate the weights on edges, namely setting them to be $\alpha f'_e - c'_e$, then those will not constitute a valid instance of the agent's demand problem, as we still need to guarantee that every subset of edges $S$ in the new graph has a valid success probability, namely $f(S) \in [0,1]$.
Moreover, to use the minimum perfect matchings in $G'$ and show that the corresponding matchings are optimal in $G$, we need to increase the weights of the edges, which are currently all non-positive. 
Thus, we change the weight function to be for every $e \in E$ and every $\alpha$:
$$
w_\alpha(e) = \alpha f_e - c_e = \alpha\left(\frac{f'_e + K}{2K|E|}\right) - \frac{c'_e}{2K|E|}
$$
Where
$$
K = \max_{e \in E} |f'_e| + \max_{e \in E} \frac{c'_e}{\alpha_2} > 0
$$
(recall that $\alpha_2 > 0$).
Where we apply the above weights on to the family of graphs $G'_n$ from Corollary \ref{cor:perfectMinMatching}, we get the super-polynomial lower bound on the number of critical values for matching-based costs and rewards.
\begin{theorem}\label{thm:superPolyCVs}
    There is a family of bipartite graphs $G_n=(V_n \cup U_n, E_n)$, which are a valid instance of the agent's demand problem for matching-based costs and rewards. The graph $G_n$ has $O(n^4)$ vertices and affine weight functions $\{w_\alpha(e) = \alpha f_e - c_e \mid e\in E_n\}$, together with a lower bound $\alpha_2 > 0$ and an upper bound $Q_n \in (\alpha_2, \infty)$, such that
    \begin{enumerate}
        \item For every $e \in E_n$, $f_e, c_e \ge 0$.
        \item For every $S \subseteq E_n$, $f(S) \in [0,1]$.
        \item For every $\alpha \in [\alpha_2,Q_n]$ and for every $e\in E_n$, $w_\alpha(e) \ge 0$
        \item There are $m = 2^{\Omega(\log^2n)}$ critical values, $0 < \alpha_2 <...< \alpha_m < Q_n$ such that for any $i\in \{2,...,m\}$, the maximum weight matching is a perfect matching. Each critical value has a distinct optimal matching.
    \end{enumerate}
\end{theorem}
\begin{proof}
    One can easily verify that the definition of $K$ guarantees that the first two conditions hold.
    For the non-negativity of the edges, consider some $\alpha > 
    \alpha_2$:
    \begin{eqnarray*}
        w_\alpha(e) &=& \alpha\left(\frac{f'_e + K}{2K|E|}\right) - \frac{c'_e}{2K|E|} \\
        &=& \frac{1}{2K|E|}(\alpha(f'_e +K) - c'_e) \\
        &\ge& \frac{1}{2K|E|}(\alpha(f'_e + \max_{e \in E} |f'_e|) + \alpha  \max_{e \in E} \frac{c'_e}{\alpha_2}- c'_e) \\
        &\ge& 0
    \end{eqnarray*}
    Because for any contract $\alpha \in [\alpha_2, Q_n]$, $w_\alpha(e) \ge 0$ for every $e \in E$, every maximum weight perfect matching is a maximum weight matching.

    For every $\alpha \in [\alpha_2, Q_n]$, a matching $M_\alpha$ is a perfect minimum matching in $(G'_n, w'_\alpha)$ if and only if it is a perfect maximum matching in $(G_n, w_\alpha)$. To see that, take any other perfect matching $M$ and observe that
    \begin{eqnarray*}
        w'_\alpha(M_\alpha) \le w'_\alpha(M) 
        &\leftrightarrow&
        \frac{1}{2K|E|}\left( |V_n|(K\alpha) -w'_\alpha(M_\alpha) \right) \ge \frac{1}{2K|E|}\left( |V_n|(K\alpha) -w'_\alpha(M) \right) \\
        &\leftrightarrow&
       \frac{1}{2K|E|} \sum_{e \in M_\alpha}  (\alpha K + \alpha f'_e - c'_e) \ge \frac{1}{2K|E|} \sum_{e \in M} (\alpha K + \alpha f'_e - c'_e) \\
        &\leftrightarrow&
        \sum_{e \in M_\alpha} \left(\alpha\left(\frac{f'_e + K}{2K|E|}\right) - \frac{c'_e}{2K|E|} \right)
        \ge \sum_{e \in M} \left(\alpha\left(\frac{f'_e + K}{2K|E|}\right) - \frac{c'_e}{2K|E|}\right) \\
        &\leftrightarrow&
        \sum_{e \in M_\alpha} (\alpha f_e - c_e )
        \ge \sum_{e \in M} (\alpha f_e - c_e ) \\
        &\leftrightarrow&
        w_\alpha(M_\alpha) \ge w_\alpha(M) 
    \end{eqnarray*}
The claim follows from Corollary \ref{cor:perfectMinMatching}.
\end{proof}

\subsection{Polynomial Number of Critical Values for Special Cases} \label{subsec:PolyMatchingBased}

In this section we show the tracability of the optimal contract for restricted cases of the matching based setting. Besides the two cases below, one can easily show that when the space of costs or rewards is of constant size, the number of critical values is small. We give the proof in \Cref{apx:ConstantRewardSpace}.

\paragraph{Pseudo-Polynomial Bound for Integer $f$ and $c$.}
When the costs and rewards take integer values, a result by Carstensen \cite{carstensen1983complexity}
implies that the optimal contract can be found in time polynomial in the number of vertices $n$, and the maximal reward $\max_{u,v} f_{u,v}$ and cost $\max_{u,v} c_{u,v}$.
\begin{theorem}[Carstensen \cite{carstensen1983complexity}]
    The number of breakpoints of an integer linear program with a parameterized objective function $\max_x \alpha f(x) - c(x)$, for $f, c \in \mathbb{Z}^n$,
    and a fixed set of constraints $Ax \le b$,
    is polynomial in $n$, $F$ and $C$. 
    Where $n$ is the dimension of $x$, $F = \max_{u,v} f_{u,v}$ and $C = \max_{u,v} c_{u,v}$.
\end{theorem}

\begin{corollary}\label{cor:MathcingQuasiPoly}
    For matching based rewards and additive costs, with $f,c \in \mathbb{Z}^{|A|}$, where $C=\max_{e\in E} c_i$ and $F=\max_{e\in E} f_i$ are both polynomial in $n$, the optimal contract can be 
    computed efficiently.
\end{corollary}

\paragraph{One-Sided Costs.}
If the costs of the edges depends only on the $U$ side of the edge, namely $c_{v,u}=c_u$ for every $(v,u) \in A$ (the case for $V$ side is symmetric), then the number of critical values can be bounded by reducing the agent's utility problem to an instance in which costs are additive and rewards are gross-substitutes, utilizing the result of \cite{dutting2022combinatorial}. 
\begin{theorem}\label{thm:oneSidedCosts}
    For a matching-based instance in which the costs are one-sided, the number of critical values is $O(n^2)$, where $n$ is the number of vertices.
\end{theorem}
Recall that for a set of edges $S\subseteq A$, when costs are one-sided the agent's utility under contract $\alpha$ is
$$
u_a(\alpha, S) = \max_{\substack{S' \subseteq S, \\ S' \text{ is a matching}}} \sum_{ (v,u) \in S'} \alpha f_{v,u}  - \sum_{(v,u) \in S} c_u
$$ 
We show that the agent's utility problem can be reformulated in terms of the vertices of $U$, with additive costs and a Rado reward function (as defined in \Cref{sec:model}).

Let $c':2^U \to \reals_{\ge 0}$ be an additive cost function. That is, for any set $T \subseteq U$, $c'(T) = \sum_{u \in T} c_u$. Let $f':2^U \to \reals$ be Rado with respect to the graph $G$ with weights $\{f_{v,u}\}_{(v,u)\in A}$ and the trivial matroid $M = (V, 2^V)$.

\begin{claim}\label{cla:RadoEquiv}
For every $\alpha \in [0,1]$, 
if $S_\alpha$ is a matching that maximizes the agent's utility among all subsets of $A$, then the set of vertices in $U$ that are matched in $S_\alpha$
maximize $\alpha f'(T) - c'(T)$ among all subsets $T \subseteq U$.
In the other direction, if $T_\alpha \subseteq U$ maximizes $\alpha f'(T) - c'(T)$ among all subsets $T \subseteq U$,
then the max weight matching in the subgraph induced by $T_\alpha$ with respect to weights $\{f_{v,u}\}_{(v,u)\in A}$, maximizes the agent's utility under contract $\alpha$.
\end{claim}

\begin{proof}
    For every matching $S \subseteq A$, denote by $U(S) \subseteq U$ the vertices from $U$ which that are matched in $S$. For every set $T \subseteq U$, denote by $A(T)$ the max weight matching in the subgraph induced by $T$ w.r.t weights $\{f_{v,u}\}_{(v,u)\in A}$. Observe that since the weights are non-negative and costs are one-sided:
    \begin{eqnarray*}
        \max_{S \subseteq A} \alpha f(S) - c(S) &=&
        \alpha f(S_\alpha) - c(S_\alpha) \\
        &=&
        \alpha \sum_{(v,u) \in S_\alpha} f_{v,u} - \sum_{(v,u) \in S_\alpha} c_{v,u} \\
        &=&
        \alpha \sum_{(v,u) \in S_\alpha} f_{v,u} - \sum_{u \in U(S_\alpha)} c_u \\
        &=&
        \alpha \sum_{(v,u) \in A(U(S_\alpha))} f_{v,u} - \sum_{u \in U(S_\alpha)} c_u \\
        &=&
        \alpha f'(U(S_\alpha)) - c'(U(S_\alpha)) \\
        &\le&
        \max_{T\subseteq U} \alpha f'(T) - c'(T)
    \end{eqnarray*}
    Where the fourth inequality follows from the optimality of $S_\alpha$ -- if $f(A(U(S_\alpha))) > f(S_\alpha)$, it will be the agent's best response as it has the same cost as $S_\alpha$. 
    On the other hand, we can use similar arguments:
    \begin{eqnarray*}
        \max_{T\subseteq U} \alpha f'(T) - c'(T)
        &=&
        \alpha f'(T_\alpha) - c'(T_\alpha) \\
        &=&
        \alpha \sum_{(v,u) \in A(T_\alpha)} f_{v,u} - \sum_{u \in T_\alpha} c_u \\
        &=&
        \alpha \sum_{(v,u) \in A(T_\alpha)} f_{v,u} - \sum_{(v,u) \in A(T_\alpha)} c_u \\
        &=&
        \alpha f(A(T_\alpha)) - c(A(T_\alpha)) \\
        &\le&
        \alpha f(S_\alpha) - c(S_\alpha)\\
    \end{eqnarray*}
    Where the third and fourth equality transition hold because $A(T_\alpha)$ is a matching.
\end{proof}

\begin{proof}[Proof of \Cref{thm:oneSidedCosts}]
    \Cref{cla:RadoEquiv} shows that under one-sided costs the agent's utility maximization problem can be formulated in terms of the vertices of $U$ and the reward and cost functions $f':2^U \to \reals_{\ge 0}$ and $c':2^U \to \reals_{\ge 0}$ described above. 
    As $c'$ is additive and $f'$ is Rado function, a sub-class of gross-substitutes, it follows from \cite{dutting2022combinatorial} that the number of critical values is at most $O(n^2)$, where $n$ refers to the number of vertices in $U$.
\end{proof}

\bibliographystyle{plain}
\bibliography{refs.bib}

\begin{thebibliography}{10}

\bibitem{AlonDLT23}
Tal Alon, Paul Duetting, Yingkai Li, and Inbal Talgam{-}Cohen.
\newblock Bayesian analysis of linear contracts.
\newblock In {\em EC 2023}, page~66, 2023.

\bibitem{alon2021contracts}
Tal Alon, Paul D{\"u}tting, and Inbal Talgam-Cohen.
\newblock Contracts with private cost per unit-of-effort.
\newblock In {\em EC 2021}, pages 52--69, 2021.

\bibitem{babaioff2006combinatorial}
Moshe Babaioff, Michal Feldman, and Noam Nisan.
\newblock Combinatorial agency.
\newblock In {\em EC 2006}, pages 18--28, 2006.

\bibitem{babaioff2006mixed}
Moshe Babaioff, Michal Feldman, and Noam Nisan.
\newblock Mixed strategies in combinatorial agency.
\newblock In {\em WINE 2006}, pages 353--364, 2006.

\bibitem{babaioff2009free}
Moshe Babaioff, Michal Feldman, and Noam Nisan.
\newblock Free-riding and free-labor in combinatorial agency.
\newblock In {\em SAGT 2009}, pages 109--121, 2009.

\bibitem{Bertelsen05}
A.~Bertelsen.
\newblock {\em Substitutes Valuations and m-concavity}.
\newblock PhD thesis, The Hebrew University, 2005.

\bibitem{carroll2015robustness}
Gabriel Carroll.
\newblock Robustness and linear contracts.
\newblock {\em Am. Econ. Rev.}, 105(2):536--563, 2015.

\bibitem{carstensen1983complexity}
Patricia~J Carstensen.
\newblock Complexity of some parametric integer and network programming
  problems.
\newblock {\em Math. Program.}, 26(1):64--75, 1983.

\bibitem{CastiglioniM021}
Matteo Castiglioni, Alberto Marchesi, and Nicola Gatti.
\newblock Bayesian agency: Linear versus tractable contracts.
\newblock In {\em EC 2021}, pages 285--286, 2021.

\bibitem{castiglioni2022designing}
Matteo Castiglioni, Alberto Marchesi, and Nicola Gatti.
\newblock Designing menus of contracts efficiently: the power of randomization.
\newblock In {\em EC 2022}, pages 705--735, 2022.

\bibitem{CastiglioniM023}
Matteo Castiglioni, Alberto Marchesi, and Nicola Gatti.
\newblock Multi-agent contract design: How to commission multiple agents with
  individual outcomes.
\newblock In {\em EC 2023}, pages 412--448, 2023.

\bibitem{dutting2022combinatorial}
Paul D{\"u}tting, Tomer Ezra, Michal Feldman, and Thomas Kesselheim.
\newblock Combinatorial contracts.
\newblock In {\em FOCS 2021}, pages 815--826, 2022.

\bibitem{duetting2022multi}
Paul D{\"{u}}tting, Tomer Ezra, Michal Feldman, and Thomas Kesselheim.
\newblock Multi-agent contracts.
\newblock In {\em STOC 2023}, pages 1311--1324, 2023.

\bibitem{DuettingGuruganeshSchneiderWang23}
Paul D\"utting, Guru Guruganesh, Jon Schneider, and Joshua Wang.
\newblock Optimal no-regret learning for one-side lipschitz functions.
\newblock In {\em ICML 2023}, 2023.
\newblock Forthcoming.

\bibitem{dutting2019simple}
Paul D{\"u}tting, Tim Roughgarden, and Inbal Talgam-Cohen.
\newblock Simple versus optimal contracts.
\newblock In {\em EC 2019}, pages 369--387, 2019.

\bibitem{dutting2021complexity}
Paul D\"utting, Tim Roughgarden, and Inbal Talgam-Cohen.
\newblock The complexity of contracts.
\newblock {\em SIAM J. Comput.}, 50(1):211--254, 2021.

\bibitem{eisner1976mathematical}
Mark~J Eisner and Dennis~G Severance.
\newblock Mathematical techniques for efficient record segmentation in large
  shared databases.
\newblock {\em J. ACM}, 23(4):619--635, 1976.

\bibitem{grossman1992analysis}
Sanford~J Grossman and Oliver~D Hart.
\newblock {\em An analysis of the principal-agent problem}.
\newblock Springer, 1992.

\bibitem{guruganesh2021contracts}
Guru Guruganesh, Jon Schneider, and Joshua~R Wang.
\newblock Contracts under moral hazard and adverse selection.
\newblock In {\em EC 2021}, pages 563--582, 2021.

\bibitem{GuruganeshSW023}
Guru Guruganesh, Jon Schneider, Joshua~R. Wang, and Junyao Zhao.
\newblock The power of menus in contract design.
\newblock In {\em EC 2023}, pages 818--848, 2023.

\bibitem{gusfield1980sensitivity}
Daniel~Mier Gusfield.
\newblock {\em Sensitivity analysis for combinatorial optimization}.
\newblock University of California, Berkeley, 1980.

\bibitem{ho2014adaptive}
Chien-Ju Ho, Aleksandrs Slivkins, and Jennifer~Wortman Vaughan.
\newblock Adaptive contract design for crowdsourcing markets: Bandit algorithms
  for repeated principal-agent problems.
\newblock In {\em EC 2014}, pages 359--376, 2014.

\bibitem{holmstrom1979moral}
Bengt Holmstr{\"o}m.
\newblock Moral hazard and observability.
\newblock {\em The Bell Journal of Economics}, pages 74--91, 1979.

\bibitem{holmstrom1991multitask}
Bengt Holmstrom and Paul Milgrom.
\newblock Multitask principal--agent analyses: Incentive contracts, asset
  ownership, and job design.
\newblock {\em J. Law Econ. Organ.}, 7:24--52, 1991.

\bibitem{iwata2002fully}
Satoru Iwata.
\newblock A fully combinatorial algorithm for submodular function minimization.
\newblock {\em J. Comb., Series B}, 84(2):203--212, 2002.

\bibitem{iwata2003faster}
Satoru Iwata.
\newblock A faster scaling algorithm for minimizing submodular functions.
\newblock {\em SIAM J. Comput.}, 32(4):833--840, 2003.

\bibitem{iwata2001combinatorial}
Satoru Iwata, Lisa Fleischer, and Satoru Fujishige.
\newblock A combinatorial strongly polynomial algorithm for minimizing
  submodular functions.
\newblock {\em J. ACM}, 48(4):761--777, 2001.

\bibitem{karp1981parametric}
Richard~M Karp and James~B Orlin.
\newblock Parametric shortest path algorithms with an application to cyclic
  staffing.
\newblock {\em Discret. Appl. Math.}, 3(1):37--45, 1981.

\bibitem{kleinberg2018delegated}
Jon Kleinberg and Robert Kleinberg.
\newblock Delegated search approximates efficient search.
\newblock In {\em EC 2018}, pages 287--302, 2018.

\bibitem{kleinberg2020classifiers}
Jon Kleinberg and Manish Raghavan.
\newblock How do classifiers induce agents to invest effort strategically?
\newblock {\em Trans. Econ. Comput.}, 8(4):1--23, 2020.

\bibitem{lehmann2001combinatorial}
Benny Lehmann, Daniel Lehmann, and Noam Nisan.
\newblock Combinatorial auctions with decreasing marginal utilities.
\newblock In {\em Proceedings of the 3rd ACM conference on Electronic
  Commerce}, pages 18--28, 2001.

\bibitem{megiddo1978combinatorial}
Nimrod Megiddo.
\newblock Combinatorial optimization with rational objective functions.
\newblock In {\em STOC 1978}, pages 1--12, 1978.

\bibitem{mulmuley2000lower}
Ketan Mulmuley and Pradyut Shah.
\newblock A lower bound for the shortest path problem.
\newblock In {\em CCC 2000}, pages 14--21, 2000.

\bibitem{nagano2007strongly}
Kiyohito Nagano.
\newblock A strongly polynomial algorithm for line search in submodular
  polyhedra.
\newblock {\em Discret. Optim.}, 4(3-4):349--359, 2007.

\bibitem{PaesLeme17}
Renato {Paes Leme}.
\newblock Gross substitutability: An algorithmic survey.
\newblock {\em Games and Economic Behavior}, 106:294--316, 2017.

\bibitem{stanley2011enumerative}
Richard~P Stanley.
\newblock Enumerative combinatorics volume 1 second edition.
\newblock {\em Cambridge studies in advanced mathematics}, 2011.

\bibitem{VuongEtAl23}
Ramiro~Deo{-}Campo Vuong, Shaddin Dughmi, Neel Patel, and Aditya Prasad.
\newblock On supermodular contracts and dense subgraphs.
\newblock {\em CoRR}, abs/2308.07473, 2023.

\bibitem{ZhuBYWJJ23}
Banghua Zhu, Stephen Bates, Zhuoran Yang, Yixin Wang, Jiantao Jiao, and
  Michael~I. Jordan.
\newblock The sample complexity of online contract design.
\newblock In {\em EC 2023}, page 1188, 2023.

\end{thebibliography}

\appendix 

\section{Matching-based Rewards and Additive Costs with a Constant-Size Reward Space}\label{apx:ConstantRewardSpace}

This section refers to a small number of different edge rewards, the case of edge costs is symmetric.
\begin{claim}\label{cla:ConstantRewardSpace}
    For matching-based rewards and additive costs,
    if the number of different values $f_{v,u}$ takes is $k = O(1)$, i.e. $|\{x \mid f_{v,u}=x, \quad v \in V, u \in U\}|=k$, then the number of critical values is $O(n^{k+1})$. Where $n$ is the number of edges.
\end{claim}
\begin{proof}
    Recall that by \Cref{obs:critVals}, for any two critical values $\alpha < \alpha'$, the best responses satisfy $f(S_\alpha) < f(S_{\alpha'})$. So the number of critical values is upper bounded by the number of different values $f(S)$ may take, where $S \subseteq A$ is a matching in $G$.
    Since $f(S) = \sum_{(v,u) \in S} f_{v,u}$, 
    for any $j \in \{1,...,n\}$, 
    the number of different values $f(S)$ may take when $|S|=j$ is upper bounded by the number of distinct multisets of size $j$ with elements taken from a set of cardinality $k$.
    This is exactly multiset coefficient, $\multiset{j}{k} = {j+k-1 \choose j}$ \cite{stanley2011enumerative}.
    The total number of different values $f$ may take is thus, 
    $\sum_{j=1}^n {j+k-1 \choose j}$, which one can easily bound with $O(n^{k+1})$.
\end{proof}

\end{document}